\documentclass{ifacconf}

\usepackage{amsmath,amsfonts}
\usepackage{algorithmicx}
\usepackage{algpseudocode}
\usepackage{algorithm}
\usepackage{array}
\usepackage[caption=false,font=normalsize,labelfont=sf,textfont=sf]{subfig}
\usepackage{textcomp}
\usepackage{stfloats}
\usepackage{url}
\usepackage{verbatim}
\usepackage{graphicx}
\usepackage{color}
\usepackage{booktabs}
\usepackage{natbib}

\theoremstyle{definition}

\theoremstyle{plain}
\newtheorem{theorem}{Theorem}

\newenvironment{proof}{\par\noindent\textit{Proof.}\ }{\hfill$\square$\par}

\begin{document}
\begin{frontmatter}

\title{Online Reinforcement Learning for Safe Gain Scheduling in Nonlinear Quadrotor Control\thanksref{footnoteinfo}} 

\thanks[footnoteinfo]{Sponsor and financial support acknowledgment
goes here. Paper titles should be written in uppercase and lowercase
letters, not all uppercase.}

\author[AME]{Muhammad Junayed Hasan Zahed}
\author[ECE]{Chieh Tsai}
\author[ECE]{Salim Hariri}
\author[AME]{Hossein Rastgoftar}

\address[AME]{Department of Aerospace and Mechanical Engineering, University of Arizona, Tucson, AZ, USA
(e-mails: \{mjhz, hrastgoftar\}@arizona.edu)}

\address[ECE]{Department of Electrical and Computer Engineering, University of Arizona, Tucson, AZ, USA
(e-mails: \{vegetableclean, hariri\}@arizona.edu)}

\begin{abstract}
This paper presents an online reinforcement-learning framework for safe gain scheduling of a nonlinear quadcopter controller. Rather than learning thrust and torque commands directly, the proposed method selects gain vectors online from a finite library of pre-certified stabilizing controllers, thereby preserving the structure of the underlying snap-based control law. Safety is enforced by restricting the policy to admissible gains that maintain forward invariance of a prescribed safe state set, while dwell-time constraints prevent excessively fast switching. To reduce the action-space dimension, translational gains are shared across spatial axes by exploiting the isotropic structure of the translational dynamics, whereas yaw gains are scheduled independently. A deep Q-network learns to adjust feedback authority according to the current flight condition, using aggressive gains during large transients and milder gains near hover. High-fidelity nonlinear simulations demonstrate accurate trajectory tracking, bounded attitude motion, reduced control effort near convergence, and stable hover regulation under online safe gain scheduling.
\end{abstract}


\end{frontmatter}

\section{Introduction}

Autonomous quadrotor stabilization and hover regulation have received growing attention as unmanned aerial vehicles (UAVs) are increasingly deployed in inspection, surveillance, transportation, and infrastructure monitoring applications \cite{mahony2012multirotor, adibfar2023review, jahani2025exploring}. Model-based nonlinear control methods, particularly flatness-based and snap-based architectures, have demonstrated strong capability for dynamically feasible tracking and reliable hover performance while preserving controller structure and interpretability \cite{lee2010geometric, mellinger2011minimum, el2023quadcopter}. However, many high-performance controllers still rely on feedback gains selected offline, which can become conservative or suboptimal across changing transient operating regimes \cite{ccopur2025tuning, hassani2024performance}. Reinforcement learning offers the possibility of online adaptation, but directly learning low-level control actions often sacrifices physical consistency, interpretability, and safety properties that are critical in safety-sensitive flight systems \cite{hsu2024reforma, dionigi2024power, xing2024bootstrapping}. These limitations motivate adaptation at the level of structured controller tuning rather than raw action generation. Accordingly, this paper develops a safe online gain-scheduling framework for snap-based quadrotor control, where reinforcement learning selects among stabilizing gain configurations in real time to improve control performance across changing operating regimes while preserving the structure, interpretability, and physical consistency of the underlying controller.

\subsection{Related Work}
Model-based quadrotor control has been extensively studied for hover stabilization and trajectory tracking because quadrotor dynamics are nonlinear, coupled, and underactuated. Flatness-based and snap-based formulations enable dynamically feasible trajectory generation and structured feedback design \cite{mellinger2011minimum, sun2022comparative, rigatos2024flatness}, while geometric control methods provide a consistent treatment of translational and attitude dynamics with strong tracking performance \cite{lee2010geometric}. These approaches are particularly attractive in safety-critical aerial systems because they preserve physical consistency and yield interpretable closed-loop behavior \cite{khalid2023control}. In practical implementations, however, controller gains are often selected or scheduled offline for nominal flight conditions, specific operating regimes, or anticipated disturbances \cite{adiguzel2024gain, khalid2023control}. Although such designs can perform well within their intended operating envelopes, they generally offer limited adaptability when flight conditions change online, especially across aggressive transients and near-hover regulation.
In contrast to the above model-based approaches, learning-based methods aim to reduce manual tuning and compensate for unmodeled dynamics in quadrotor flight. Many such methods retain a stabilizing feedback loop while augmenting it with learned feedforward or adaptive components. Prior work has used deep neural networks for trajectory tracking \cite{li2017deep}, inversion-based learning for aggressive multirotor control \cite{spitzer2019inverting}, and reinforcement learning for agile or low-level UAV flight control \cite{koch2019reinforcement}. However, most RL-based flight controllers still learn continuous low-level control actions or generic policy representations, which can reduce interpretability and make it difficult to preserve the structure of established model-based controllers.

Safe reinforcement learning addresses these issues through several representative directions. Barrier-function-based approaches enforce safety through filtering or shielding \cite{zhang2021model, emam2022safe}, Lyapunov-based approaches incorporate stability-oriented certificates during learning \cite{chow2018lyapunov}, and constrained policy optimization introduces explicit safety constraints into policy updates \cite{achiam2017constrained}. While these methods provide important mechanisms for safe policy improvement, they are typically developed for generic learned policies rather than structured controller tuning. In contrast, our approach uses reinforcement learning to perform online gain scheduling within a snap-based control architecture, so that adaptation occurs at the level of structured feedback gains rather than raw control generation. This preserves controller structure and interpretability while enabling real-time adjustment of feedback authority across changing flight regimes.

\subsection{Contributions}
This paper develops a safe and adaptive control framework for quadrotor hover regulation that integrates model-based snap control with online learning-based gain scheduling. The main contributions are as follows.

\begin{itemize}
\item We formulate quadrotor hover regulation as a two-level problem consisting of: (i) construction of a gain-parameterized feedback law that guarantees forward invariance of an admissible state set, and (ii) online selection of gains from an admissible discrete gain library to optimize closed-loop performance.

\item We develop a forward-invariant control architecture in which safety is enforced by design through an admissible gain set $\mathcal{K}_{\mathrm{adm}}$. This separates safety from performance optimization and ensures that the learning policy operates only over controllers that preserve safe closed-loop evolution.

\item We construct an online learning framework for gain scheduling in which the policy selects feedback gains from a finite subset $\bar{\mathcal{K}}_{\mathrm{adm}} \subset \mathcal{K}_{\mathrm{adm}}$, rather than directly synthesizing control inputs. This preserves the structure, interpretability, and physical consistency of the model-based controller.

\item We provide a unified architecture that connects reference generation, error feedback, nonlinear inversion, forward invariance guarantees, and discrete-time learning, enabling safe online adaptation during transient maneuvers toward hover.
\end{itemize}
\subsection{Outline}
This paper is organized as follows. The problem statement is presented in Section~\ref{Problem Statement}. A feedback control law that guarantees forward invariance of the quadrotor dynamics is developed in Section~\ref{Quadcopter Model}. The safe online learning framework for gain scheduling is introduced in Section~\ref{Safe Learning}. Simulation results are presented in Section~\ref{Results}, followed by concluding remarks in Section~\ref{Conclusion}.
\section{Problem Statement}\label{Problem Statement}
We study the problem of safe, adaptive trajectory tracking for a quadrotor and seek to develop a control architecture that integrates model-based feedback with online learning. The objective is to drive the vehicle to a desired hover condition while guaranteeing safe closed-loop operation.

Using the quadrotor model in \cite{el2023quadcopter, Rastgoftar2021TCNS}, the dynamics are given by
\begin{equation}\label{orginaldyn}
\begin{cases}
\dot{\mathbf{r}} = \mathbf{v}, \\[4pt]
\dot{\mathbf{v}}
=
-\,g\hat{\mathbf{e}}_3
+\dfrac{mg+T}{m}\mathbf{R}(\boldsymbol{\eta})\hat{\mathbf{e}}_3, \\[8pt]
\dot{\boldsymbol{\eta}} = \mathbf{E}(\phi,\theta)\boldsymbol{\omega}, \\[8pt]
\mathbf{I}\dot{\boldsymbol{\omega}}
+
\boldsymbol{\omega}\times(\mathbf{I}\boldsymbol{\omega})
=
\boldsymbol{\tau}, \\[8pt]
\dot T = \dot{T}, \\[4pt]
\ddot T = u_T,
\end{cases}
\end{equation}
where $\hat{\mathbf{e}}_3=[0~0~1]^\top$, $\mathbf{r},\mathbf{v}\in\mathbb{R}^3$ denote inertial position and velocity, 
$\boldsymbol{\eta}=(\phi,\theta,\psi)$ are the 3-2-1 Euler angles, 
$\boldsymbol{\omega}=(p,q,r)$ is the body angular velocity, 
$T$ represents thrust deviation from hover, 
$\boldsymbol{\tau}=(\tau_\phi,\tau_\theta,\tau_\psi)^\top$ is the control torque, 
and $\mathbf{E}(\phi,\theta)$ is the Euler kinematic matrix.

Define the state and input as
\begin{equation}
\begin{aligned}
\mathbf{x}
&=
\begin{bmatrix}
\mathbf{r}^\top & \mathbf{v}^\top & \boldsymbol{\eta}^\top & \dot{\boldsymbol{\eta}}^\top & T & \dot{T}
\end{bmatrix}^\top
\in\mathbb{R}^{14},\\
\mathbf{u}
&=
\begin{bmatrix}
u_T & \mathbf{u}_\eta^\top
\end{bmatrix}^\top
\in\mathbb{R}^{4},
\end{aligned}
\end{equation}
where $\mathbf{u}_\eta \in \mathbb{R}^3$ denotes the angular acceleration input in Euler coordinates.

Using the relation $\boldsymbol{\omega}=\mathbf{E}^{-1}(\phi,\theta)\dot{\boldsymbol{\eta}}$, the dynamics can be expressed in control-affine form as
\begin{equation}\label{dyn0}
\dot{\mathbf{x}}=\mathbf{f}_0(\mathbf{x})+\mathbf{G}_0(\mathbf{x})\mathbf{u},\qquad \mathbf{x}\left(0\right)=\mathbf{x}_0,
\end{equation}
where $\mathbf{x}_0\in \mathbb{R}^{14}$ is the initial state, 
\begin{equation}
\mathbf{f}_0(\mathbf{x})=
\begin{bmatrix}
\mathbf{v} \\[4pt]
-\,g\hat{\mathbf{e}}_3+\dfrac{mg+T}{m}\mathbf{R}(\boldsymbol{\eta})\hat{\mathbf{e}}_3 \\[8pt]
\dot{\boldsymbol{\eta}} \\[4pt]
\mathbf{0}_3 \\[6pt]
\dot{T} \\[4pt]
0
\end{bmatrix}.
\end{equation}
\begin{equation}
\mathbf{G}_0\!\left(\mathbf{x}\right)=
\begin{bmatrix}
\mathbf{0}_{9\times 1}&\mathbf{0}_{9\times 3}\\
\mathbf{0}_{3\times 1} & \mathbf{I}_3\\
0 & \mathbf{0}_{1\times 3}\\
1 & \mathbf{0}_{1\times 3}
\end{bmatrix}.
\end{equation}

Given the nonlinear quadrotor dynamics in \eqref{orginaldyn}--\eqref{dyn0}, we aim to achieve safe and optimal hover by addressing the following two coupled problems:

\noindent\textbf{Problem 1 (Forward Invariance Guarantee):}
Let $\mathcal{X}_{\mathrm{adm}} \subset \mathbb{R}^{14}$ denote the admissible state set. 
Design a feedback law
\begin{equation}
\mathbf{u} = \mathbf{h}(\mathbf{x}, \mathbf{k}),
\end{equation}
such that $\mathcal{X}_{\mathrm{adm}}$ is forward invariant under the closed-loop dynamics, i.e.,
\begin{equation}
\mathbf{x}(0) \in \mathcal{X}_{\mathrm{adm}}
\;\Longrightarrow\;
\mathbf{x}(t) \in \mathcal{X}_{\mathrm{adm}}, \quad \forall t \ge 0.
\end{equation}
This is achieved by restricting the gain to an admissible set
\begin{equation}
\mathbf{k} \in \mathcal{K}_{\mathrm{adm}} \subset \mathbb{R}^{14},
\end{equation}
constructed such that the above invariance condition holds for all $\mathbf{k} \in \mathcal{K}_{\mathrm{adm}}$.

\noindent\textbf{Problem 2 (Safe Online Gain Scheduling):}
Under the feedback law $\mathbf{u}=\mathbf{h}(\mathbf{x},\mathbf{k})$, the discretized closed-loop dynamics can be written as
\begin{equation}\label{discretetime}
\mathbf{x}_{t+1}=\mathbf{\Phi}(\mathbf{x}_t,\mathbf{k}_t), \qquad t=0,1,2,\ldots,
\end{equation}
where $\mathbf{x}_t \in \mathcal{X}_{\mathrm{adm}}$, $\mathbf{k}_t \in \mathcal{K}_{\mathrm{adm}}$, and $\mathbf{x}_{t+1} \in \mathcal{X}_{\mathrm{adm}}$. The transition map $\mathbf{\Phi}$ is unknown and is accessed only through interaction with the system.

By discretizing the admissible gain set $\mathcal{K}_{\mathrm{adm}}$, we obtain a finite set 
$\bar{\mathcal{K}}_{\mathrm{adm}} \subset \mathcal{K}_{\mathrm{adm}}$ of stabilizing control gains. 
We then design an online learning policy
\begin{equation}
\mathbf{k}_t = \pi(\mathbf{x}_t),
\end{equation}
such that $\mathbf{k}_t \in \bar{\mathcal{K}}_{\mathrm{adm}}$ for all $t$, while optimizing closed-loop performance over time.

\section{Forward Invariant Control}\label{Quadcopter Model}

This section constructs the feedback law posed in Problem 1 and derives the control-friendly discrete-time dynamics used in Problem 2. The objective is to regulate the quadrotor to a desired hover condition while ensuring that the closed-loop state remains in the admissible set $\mathcal{X}_{\mathrm{adm}}$ for all time. 

By adopting the snap-based control framework of \cite{el2023quadcopter}, the quadrotor is regulated toward the hover equilibrium
\[
\mathbf{x}^*=
\begin{bmatrix}
\mathbf{r}_I^*\\
\mathbf{0}_{11\times 1}
\end{bmatrix},
\]
where $\mathbf{r}_I^* \in \mathbb{R}^3$ denotes the desired inertial-frame hover position.

To ensure a smooth and dynamically feasible transition from the initial position $\mathbf{r}_0$ to the target $\mathbf{r}_I^*$, we construct a reference trajectory of the form
\begin{equation}\label{rd}
\mathbf{r}_d(t)=
\begin{cases}
(1-\beta(t))\mathbf{r}_0+\beta(t)\mathbf{r}_I^*, & t\in[0,T_f],\\[4pt]
\mathbf{r}_I^*, & t\ge T_f,
\end{cases}
\end{equation}
where $\beta:[0,T_f]\to[0,1]$ is a quintic polynomial satisfying boundary conditions that ensure
\[
\mathbf{r}_d(t) \in C^4,
\]
i.e., continuity of position and its derivatives up to snap at $t=T_f$.

This construction guarantees a smooth transition that is consistent with the differential flatness and snap-based control requirements of the quadrotor dynamics, enabling accurate tracking and stable convergence to the desired hover equilibrium.

Define the tracking-error state
\[
\mathbf{z}(\mathbf{x})
=
\begin{bmatrix}
\mathbf{e}_r^\top &
\mathbf{e}_v^\top &
\mathbf{e}_a^\top &
\mathbf{e}_j^\top &
\psi &
\dot{\psi}
\end{bmatrix}^\top
\in\mathbb{R}^{14},
\]
where
\[
\mathbf{e}_r=\mathbf{r}-\mathbf{r}_d,~
\mathbf{e}_v=\mathbf{v}-\dot{\mathbf{r}}_d,~
\mathbf{e}_a=\mathbf{a}-\ddot{\mathbf{r}}_d,~
\mathbf{e}_j=\mathbf{j}-\dddot{\mathbf{r}}_d.
\]

Let the external input be
\[
\mathbf{s}
=
\begin{bmatrix}
\mathbf{s}_r\\
s_\psi
\end{bmatrix}
=
\begin{bmatrix}
\dot{\mathbf{j}}\\
\ddot{\psi}
\end{bmatrix}
\in\mathbb{R}^{4}.
\]
Then the external dynamics take the form
\begin{equation}
\dot{\mathbf{z}}
=
\mathbf{A}_{\mathrm{EXT}}\mathbf{z}
+
\mathbf{B}_{\mathrm{EXT}}\mathbf{s}
+
\mathbf{E}_{\mathrm{EXT}}\mathbf{r}_d^{(4)}(t),
\end{equation}
where
\begin{equation}
\mathbf{A}_{\mathrm{EXT}}=
\begin{bmatrix}
\mathbf{0}_{3\times 3} & \mathbf{I}_3 & \mathbf{0}_{3\times 3} & \mathbf{0}_{3\times 3} & \mathbf{0}_{3\times 1} & \mathbf{0}_{3\times 1}\\
\mathbf{0}_{3\times 3} & \mathbf{0}_{3\times 3} & \mathbf{I}_3 & \mathbf{0}_{3\times 3} & \mathbf{0}_{3\times 1} & \mathbf{0}_{3\times 1}\\
\mathbf{0}_{3\times 3} & \mathbf{0}_{3\times 3} & \mathbf{0}_{3\times 3} & \mathbf{I}_3 & \mathbf{0}_{3\times 1} & \mathbf{0}_{3\times 1}\\
\mathbf{0}_{3\times 3} & \mathbf{0}_{3\times 3} & \mathbf{0}_{3\times 3} & \mathbf{0}_{3\times 3} & \mathbf{0}_{3\times 1} & \mathbf{0}_{3\times 1}\\
\mathbf{0}_{1\times 3} & \mathbf{0}_{1\times 3} & \mathbf{0}_{1\times 3} & \mathbf{0}_{1\times 3} & 0 & 1\\
\mathbf{0}_{1\times 3} & \mathbf{0}_{1\times 3} & \mathbf{0}_{1\times 3} & \mathbf{0}_{1\times 3} & 0 & 0
\end{bmatrix},
\end{equation}
\begin{equation}
\mathbf{B}_{\mathrm{EXT}}=
\begin{bmatrix}
\mathbf{0}_{3\times 3} & \mathbf{0}_{3\times 1}\\
\mathbf{0}_{3\times 3} & \mathbf{0}_{3\times 1}\\
\mathbf{0}_{3\times 3} & \mathbf{0}_{3\times 1}\\
\mathbf{I}_3 & \mathbf{0}_{3\times 1}\\
\mathbf{0}_{1\times 3} & 0\\
\mathbf{0}_{1\times 3} & 1
\end{bmatrix},
\end{equation}
and
\begin{equation}
\mathbf{E}_{\mathrm{EXT}}=
\begin{bmatrix}
\mathbf{0}_{3\times 3}\\
\mathbf{0}_{3\times 3}\\
\mathbf{0}_{3\times 3}\\
-\mathbf{I}_3\\
\mathbf{0}_{1\times 3}\\
\mathbf{0}_{1\times 3}
\end{bmatrix}
\in \mathbb{R}^{14\times 3}.
\end{equation}

We choose the gain-parameterized feedback law
\begin{equation}\label{eq:s_feedback}
\resizebox{0.99\hsize}{!}{%
$
\mathbf{s}
=
\begin{bmatrix}
\mathbf{K}_j(\dddot{\mathbf{r}}_d-\mathbf{j})
+\mathbf{K}_a(\ddot{\mathbf{r}}_d-\mathbf{a})
+\mathbf{K}_v(\dot{\mathbf{r}}_d-\mathbf{v})
+\mathbf{K}_p(\mathbf{r}_d-\mathbf{r})\\[4pt]
-k_{13}\dot{\psi}-k_{14}\psi
\end{bmatrix},
$
}
\end{equation}
with diagonal gain matrices
\[
\begin{split}
\mathbf{K}_j=&\mathrm{diag}(k_1,k_2,k_3),\\
\mathbf{K}_a=&\mathrm{diag}(k_4,k_5,k_6), \\
\mathbf{K}_v=&\mathrm{diag}(k_7,k_8,k_9), \\
\mathbf{K}_p=&\mathrm{diag}(k_{10},k_{11},k_{12}),
\end{split}
\]
and yaw gains $k_{13},k_{14}>0$. Collect these gains in
\[
\mathbf{k}=
\begin{bmatrix}
k_1&\cdots&k_{14}
\end{bmatrix}^\top \in \mathbb{R}^{14}.
\]

\noindent \textbf{Closed-loop structure, stability, and Forward Invariance Guarantee.}
Substituting \eqref{eq:s_feedback} into the external dynamics yields a linear time-varying system in the tracking error coordinates:
\begin{equation}
\dot{\mathbf{z}}
=
\mathbf{A}_{\mathrm{cl}}(\mathbf{k})\,\mathbf{z}
+
\mathbf{E}_{\mathrm{EXT}}\mathbf{r}_d^{(4)}(t),
\end{equation}
where $\mathbf{A}_{\mathrm{cl}}(\mathbf{k})=\mathbf{A}_{\mathrm{EXT}}-\mathbf{B}_{\mathrm{EXT}}\mathbf{K}(\mathbf{k})$, and $\mathbf{K}(\mathbf{k})\in\mathbb{R}^{4\times 14}$ is the structured feedback gain matrix induced by \eqref{eq:s_feedback}, defined as
\begin{equation}
\mathbf{K}(\mathbf{k})=
\begin{bmatrix}
\mathbf{K}_p & \mathbf{K}_v & \mathbf{K}_a & \mathbf{K}_j & \mathbf{0}_{3\times 1}& \mathbf{0}_{3\times 1}\\[4pt]
\mathbf{0}_{1\times 3} &\mathbf{0}_{1\times 3} &\mathbf{0}_{1\times 3} &\mathbf{0}_{1\times 3} & k_{14} &k_{13}
\end{bmatrix}.
\end{equation}

\begin{theorem}[Forward invariance under gain-parameterized feedback]
\label{thm:forward_invariance}
Consider the closed-loop error dynamics
\begin{equation}
\dot{\mathbf z}
=
\mathbf A_{\mathrm{cl}}(\mathbf k)\mathbf z
+
\mathbf E_{\mathrm{EXT}}\mathbf r_d^{(4)}(t),
\end{equation}
where $\mathbf A_{\mathrm{cl}}(\mathbf k)$ is Hurwitz for a given $\mathbf k\in\mathcal K_{\mathrm{adm}}$, and $\mathbf r_d^{(4)}(t)$ is bounded with
\[
\bar r := \|\mathbf r_d^{(4)}\|_\infty < \infty.
\]
Suppose there exist matrices $\mathbf P\succ 0$ and $\mathbf Q\succ 0$ such that
\begin{equation}
\mathbf A_{\mathrm{cl}}^\top(\mathbf k)\mathbf P
+
\mathbf P\mathbf A_{\mathrm{cl}}(\mathbf k)
\le
-\mathbf Q.
\end{equation}
Define
\[
V(\mathbf z)=\mathbf z^\top \mathbf P \mathbf z,
\]
\[
\mathcal Z_\rho
=
\left\{
\mathbf z\in\mathbb R^{14}\;\middle|\; V(\mathbf z)\le \rho
\right\}.
\]
Then, there exists $\rho^\star>0$ such that, for every $\rho\ge \rho^\star$, the set $\mathcal Z_\rho$ is forward invariant. Moreover, the origin is uniformly ultimately bounded, with ultimate bound proportional to $\bar r$.
\end{theorem}

\begin{proof}
Along trajectories,
\begin{align}
\dot V(\mathbf z)
&=
\mathbf z^\top
\left(
\mathbf A_{\mathrm{cl}}^\top \mathbf P + \mathbf P \mathbf A_{\mathrm{cl}}
\right)\mathbf z
+
2\mathbf z^\top \mathbf P \mathbf E_{\mathrm{EXT}}\mathbf r_d^{(4)}(t) \\
&\le
-\lambda_{\min}(\mathbf Q)\|\mathbf z\|^2
+
2\|\mathbf P\mathbf E_{\mathrm{EXT}}\|\,\|\mathbf z\|\,\bar r,
\end{align}
where we used $\mathbf A_{\mathrm{cl}}^\top \mathbf P + \mathbf P \mathbf A_{\mathrm{cl}} \le -\mathbf Q$,
the bound $\mathbf z^\top \mathbf Q \mathbf z \ge \lambda_{\min}(\mathbf Q)\|\mathbf z\|^2$,
and the inequality
\[|\mathbf z^\top \mathbf P \mathbf E_{\mathrm{EXT}} \mathbf r|
\le \|\mathbf z\| \|\mathbf P \mathbf E_{\mathrm{EXT}}\| \|\mathbf r\|\]
together with $\|\mathbf r_d^{(4)}(t)\|\le \bar r$.
Let
\[
\alpha := \lambda_{\min}(\mathbf Q),
\]
\[
c := 2\|\mathbf P\mathbf E_{\mathrm{EXT}}\|,
\]
\[
R := \frac{c}{\alpha}\bar r.
\]
Then,
\[
\dot V(\mathbf z) < 0,
\qquad \forall\, \mathbf z \in \mathbb{R}^{14} \text{ with } \|\mathbf z\| > R.
\]
Now, if $\mathbf z\in \partial \mathcal Z_\rho$, where $\partial \mathcal Z_\rho := \{\mathbf z \in \mathbb{R}^{14} \mid V(\mathbf z)=\rho\}$ denotes the boundary of the sublevel set. Then
\[
\rho = \mathbf z^\top \mathbf P \mathbf z \le \lambda_{\max}(\mathbf P)\|\mathbf z\|^2,
\]
so
\[
\|\mathbf z\| \ge \sqrt{\frac{\rho}{\lambda_{\max}(\mathbf P)}}.
\]
Hence, if
\[
\rho^\star := \lambda_{\max}(\mathbf P)R^2,
\]
then for every $\rho\ge \rho^\star$ and every $\mathbf z\in\partial \mathcal Z_\rho$,
\[
\|\mathbf z\| \ge R,
\]
which implies $\dot V(\mathbf z)\le 0$ on $\partial \mathcal Z_\rho$. Therefore, $\mathcal Z_\rho$ is forward invariant for all $\rho\ge \rho^\star$.

Since $\dot V(\mathbf z)<0$ outside the ball $\{\mathbf z:\|\mathbf z\|\le R\}$, every trajectory enters and remains in a compact neighborhood of the origin. Thus, $\mathbf z(t)$ is uniformly ultimately bounded, and the ultimate bound scales linearly with $\bar r=\|\mathbf r_d^{(4)}\|_\infty$.
\end{proof}

\noindent \textbf{Invariant set construction.}
Let the admissible state set be defined as
\[
\mathcal{X}_{\mathrm{adm}}
=
\{\mathbf{x} \mid \mathbf{z}(\mathbf{x}) \in \mathcal{Z}_{\mathrm{adm}}\},
\]
where $\mathcal{Z}_{\mathrm{adm}}$ is a compact set of the form
\[
\mathcal{Z}_{\mathrm{adm}}=
\{\mathbf{z}\in\mathbb{R}^{14} \mid V(\mathbf{z}) \le \rho\},
\]
with $\rho>0$ chosen such that all corresponding physical constraints (e.g., position bounds, velocity limits, thrust feasibility, and attitude constraints) are satisfied.

By standard Lyapunov arguments, $\mathcal{Z}_{\mathrm{adm}}$ is forward invariant if
\[
\dot V(\mathbf{z}) \le 0
\quad \text{on} \quad
\partial \mathcal{Z}_{\mathrm{adm}}.
\]
This condition is guaranteed provided that the gains $\mathbf{k}$ are selected such that
\[
\lambda_{\min}(\mathbf{Q})\rho
\ge
c^2 \|\mathbf{r}_d^{(4)}\|_\infty^2.
\]

\noindent \textbf{Admissible gain set.}
We therefore define the admissible gain set
\[
\mathcal{K}_{\mathrm{adm}}
=
\left\{
\mathbf{k} \in \mathbb{R}^{14}
\;\middle|\;
\begin{array}{l}
\mathbf{A}_{\mathrm{cl}}(\mathbf{k}) \text{ is Hurwitz},\\
\mathcal{Z}_{\mathrm{adm}} \text{ is forward invariant}
\end{array}
\right\}.
\]
By construction, for every $\mathbf{k}\in\mathcal{K}_{\mathrm{adm}}$, the closed-loop system satisfies
\[
\mathbf{x}(0)\in\mathcal{X}_{\mathrm{adm}}
\;\Longrightarrow\;
\mathbf{x}(t)\in\mathcal{X}_{\mathrm{adm}}, \qquad \forall t\ge 0.
\]
\noindent\textbf{Relation between external input and virtual input.}
Define
\[
\hat{\mathbf{k}}_b=\mathbf{R}(\boldsymbol{\eta})\hat{\mathbf{e}}_3
\]
as the thrust direction. By differentiating the translational dynamics twice and using the kinematic relations, the external input can be written as
\begin{equation}
\mathbf{s}=\mathbf{M}(\mathbf{x})\mathbf{u}+\mathbf{n}(\mathbf{x}),
\label{su}
\end{equation}
where $\mathbf{M}(\mathbf{x})\in\mathbb{R}^{4\times 4}$ and $\mathbf{n}(\mathbf{x})\in\mathbb{R}^{4}$ are state-dependent terms induced by the quadrotor dynamics. Hence,
\begin{equation}
\mathbf{u}
=
\mathbf{M}^{-1}(\mathbf{x})\bigl(\mathbf{s}-\mathbf{n}(\mathbf{x})\bigr).
\end{equation}
Combining this with \eqref{eq:s_feedback} yields the feedback law
\begin{equation}
\mathbf{u}=\mathbf{h}(\mathbf{x},\mathbf{k}),
\end{equation}
which is exactly the controller introduced in Problem 1.

\noindent \textbf{Control-Friendly Discrete-Time Dynamics.} 
Since the external input depends linearly on the gain vector, \eqref{eq:s_feedback} can be written compactly as
\begin{equation}
\mathbf{s}=\mathbf{H}(\mathbf{x})\mathbf{k},
\end{equation}
where $\mathbf{H}:\mathbb{R}^{14}\to\mathbb{R}^{4\times 14}$. Substituting into \eqref{su} gives
\begin{equation}
\mathbf{u}
=
\mathbf{M}^{-1}(\mathbf{x})\bigl(\mathbf{H}(\mathbf{x})\mathbf{k}-\mathbf{n}(\mathbf{x})\bigr).
\end{equation}
Therefore, the closed-loop continuous-time dynamics can be written as
\begin{equation}
\dot{\mathbf{x}}
=
\mathbf{f}(\mathbf{x})
+
\mathbf{G}(\mathbf{x})\mathbf{k},
\label{eq:control_affine}
\end{equation}
where
\begin{equation}
    \mathbf{f}(\mathbf{x})
=
\mathbf{f}_0(\mathbf{x})
-
\mathbf{G}_0(\mathbf{x})\mathbf{M}^{-1}(\mathbf{x})\mathbf{n}(\mathbf{x}),
\end{equation}
\begin{equation}
    \mathbf{G}(\mathbf{x})
=
\mathbf{G}_0(\mathbf{x})\mathbf{M}^{-1}(\mathbf{x})\mathbf{H}(\mathbf{x}).
\end{equation}
Under zero-order hold with sampling time $\Delta t$, the closed-loop dynamics admit the discrete-time representation, given by Eq. \eqref{discretetime}, 
where $\Phi$ is obtained via fourth-order Runge--Kutta integration of \eqref{eq:control_affine}. This discrete-time model serves as the basis for Problem~2 and for constructing the reinforcement learning environment.

\section{Reinforcement Learning Formulation}\label{Safe Learning}

The gain-scheduled closed-loop system is modeled as a discrete-time controlled system \eqref{discretetime}
where $\mathbf{x}_t \in \mathcal{X}_{\mathrm{adm}}$ and $\mathbf{k}_t \in \bar{\mathcal{K}}_{\mathrm{adm}}\subset {\mathcal{K}}_{\mathrm{adm}}$. The transition map $\Phi$ is unknown and is accessed through interaction with the system. The admissible gain table specifies a family of stabilizing controllers, while the bounds $k^{i}_{\min}$ and $k^{i}_{\max}$ reported in Table~\ref{tab:gain_bounds} represent, respectively, the component-wise minimum and maximum values across that table. At each time step, the learning agent selects a gain index corresponding to $\mathbf{k}_t \in \bar{\mathcal{K}}_{\mathrm{adm}}$. To prevent arbitrarily fast switching between controllers, a dwell-time constraint is enforced such that each selected gain is held constant for a fixed number of time steps.

\subsection{Reward Function}

The instantaneous reward is defined as
\begin{equation}
\begin{aligned}
r_t = -\Big(
& w_r \|\mathbf{e}_r(t)\|^2
+ w_v \|\mathbf{e}_v(t)\|^2 \\
& + w_\eta \|\boldsymbol{\eta}(t)\|^2
+ w_\omega \|\boldsymbol{\omega}(t)\|^2
\Big) \\
& - w_u \left( \ddot{T}(t)^2 + \|\boldsymbol{\tau}(t)\|^2 \right)
- w_s \,\mathbb{I}(\mathbf{k}_t \neq \mathbf{k}_{t-1}),
\end{aligned}
\end{equation}
where $\boldsymbol{\tau}(t)$ denotes the control torque and $\mathbb{I}(\cdot)$ is the indicator function.

Additionally, a large terminal penalty is imposed if safety constraints are violated, e.g.,
\begin{equation}
r_t \leftarrow r_t - \rho_{\mathrm{fail}}, \qquad \rho_{\mathrm{fail}} \gg 1,
\end{equation}
whenever the state exits admissible bounds or numerical instability is detected.

This reward formulation penalizes tracking error, attitude deviation, angular velocity, control effort, and excessive switching, while strongly discouraging unsafe trajectories.

\subsection{Online Q-Learning over Gain Library}

Let $\bar{\mathcal{K}}_{\mathrm{adm}} = \{\mathbf{k}^{(1)},\dots,\mathbf{k}^{(N)}\}$ denote the finite gain set. The problem reduces to selecting one gain at each time step.

Define the action-value function
\begin{equation}
Q(\mathbf{x},\mathbf{k}) = \mathbb{E}\left[\sum_{t=0}^{\infty} \gamma^t r_t \,\bigg|\, \mathbf{x}_0=\mathbf{x}, \mathbf{k}_0=\mathbf{k} \right].
\end{equation}

The gain selection policy is implemented using an $\varepsilon$-greedy strategy:
\begin{equation}
\mathbf{k}_t =
\begin{cases}
\text{random element of } \bar{\mathcal{K}}_{\mathrm{adm}}, & \text{with probability } \varepsilon, \\
\displaystyle \arg\max_{\mathbf{k} \in \bar{\mathcal{K}}_{\mathrm{adm}}} Q_\theta(\mathbf{x}_t,\mathbf{k}), & \text{otherwise}.
\end{cases}
\end{equation}

The action-value function is approximated by a neural network and updated online via temporal-difference learning:
\begin{equation}
\mathcal{L}(\theta)
=
\left(
Q_\theta(\mathbf{x}_t,\mathbf{k}_t)
-
\big[
r_t + \gamma (1-d_t)\max_{\mathbf{k}'} Q_\theta(\mathbf{x}_{t+1},\mathbf{k}')
\big]
\right)^2,
\end{equation}
where $\gamma \in (0,1)$ is the discount factor and $d_t$ indicates episode termination.

\begin{table*}[t]
\centering
\caption{Admissible gain bounds, $\bar{\mathcal{K}}_{\mathrm{adm}}$ derived from the stabilizing gain table used by the online training environment.}
\label{tab:gain_bounds}
\renewcommand{\arraystretch}{1.15}
\resizebox{\textwidth}{!}{%
\begin{tabular}{lcccccccccccccc}
\toprule
 & $k^{(1)}$ & $k^{(2)}$ & $k^{(3)}$ & $k^{(4)}$ & $k^{(5)}$ & $k^{(6)}$ & $k^{(7)}$ & $k^{(8)}$ & $k^{(9)}$ & $k^{(10)}$ & $k^{(11)}$ & $k^{(12)}$ & $k^{(13)}$ & $k^{(14)}$ \\
\midrule
$k^{i}_{\min}$ & 9.8304 & 24.1920 & 49.1520 & 25.6000 & 47.6160 & 78.8480 & 22.4000 & 32.9600 & 45.4400 & 8.0000 & 9.6000 & 11.2000 & 12.0000 & 8.0000 \\
$k^{i}_{\max}$ & 49.7664 & 122.4720 & 248.8320 & 86.4000 & 160.7040 & 266.1120 & 50.4000 & 74.1600 & 102.2400 & 12.0000 & 14.4000 & 16.8000 & 32.0000 & 12.0000 \\
\bottomrule
\end{tabular}%
}
\end{table*}

\section{Simulation Results}\label{Results}

We evaluate the proposed DQN-based gain-scheduling framework in a high-fidelity nonlinear quadcopter simulation using a ZYX Euler-angle attitude representation and thrust/torque actuation. 
The quadcopter parameters are chosen as $m=1.5$ kg, $g=9.81$ m/s$^2$, and $\mathbf{I}=\mathrm{diag}(0.02,\,0.02,\,0.04)$ kg\,m$^2$. Numerical integration is carried out with a fixed step size of $\Delta t=0.01$ s, and each episode spans $10$ s. The desired trajectory is generated a smooth quintic time-scaling over $t\in[0,T_f]$ with $T_f=5$ s. For $t>T_f$, the desired position is held at $\mathbf{r}_d(T_f)$, while the desired velocity, acceleration, jerk, and snap are all set to zero.

\begin{figure}[!h]
\centering
\includegraphics[width=0.99\linewidth]{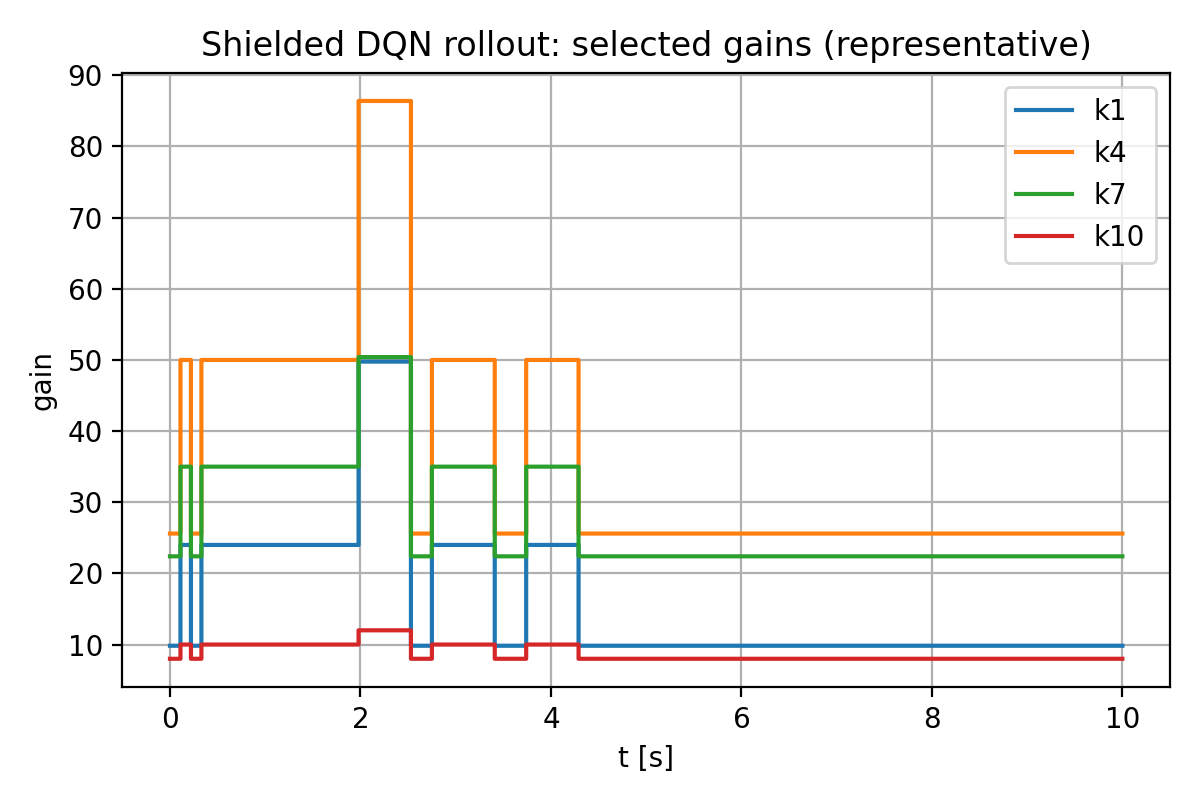}
\vspace{-1cm}
\caption{Shielded DQN rollout: selected translational feedback gains shared across axes. The policy increases gains during the initial transient and reduces them as tracking errors diminish.}
\label{fig:dqn_gains}
\end{figure}

\begin{figure}[!h]
\centering
\includegraphics[width=0.99\linewidth]{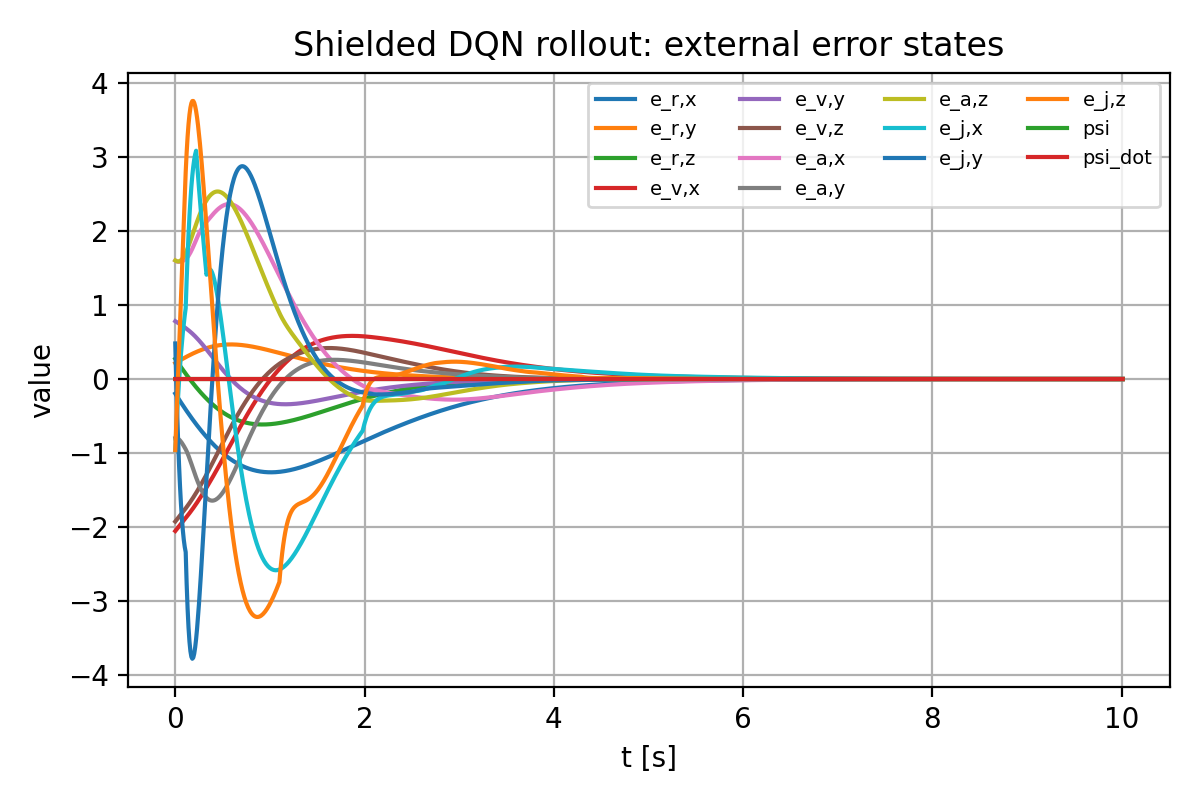}
\caption{Shielded DQN rollout: external error states. Translational error components (position, velocity, acceleration, and jerk) converge toward the origin while the yaw channel remains regulated, demonstrating stable closed-loop behavior under learned gain scheduling.}
\label{fig:dqn_states}
\end{figure}

The DQN observation is obtained by augmenting the 14-dimensional physical state with the phase variable $\min(t/T_f,1)$, resulting in a 15-dimensional input vector. At each decision step, the agent selects a gain vector from a finite library of pre-certified stabilizing gains. To reduce the policy dimension while exploiting the symmetry of the translational dynamics, the gains are tied across the $x$-, $y$-, and $z$-channels at each derivative level, whereas the yaw gains are scheduled separately. A dwell-time constraint is also enforced so that each selected gain remains active for a prescribed number of steps, thereby preventing overly rapid switching.

\begin{figure}[!h]
\centering
\includegraphics[width=0.95\linewidth]{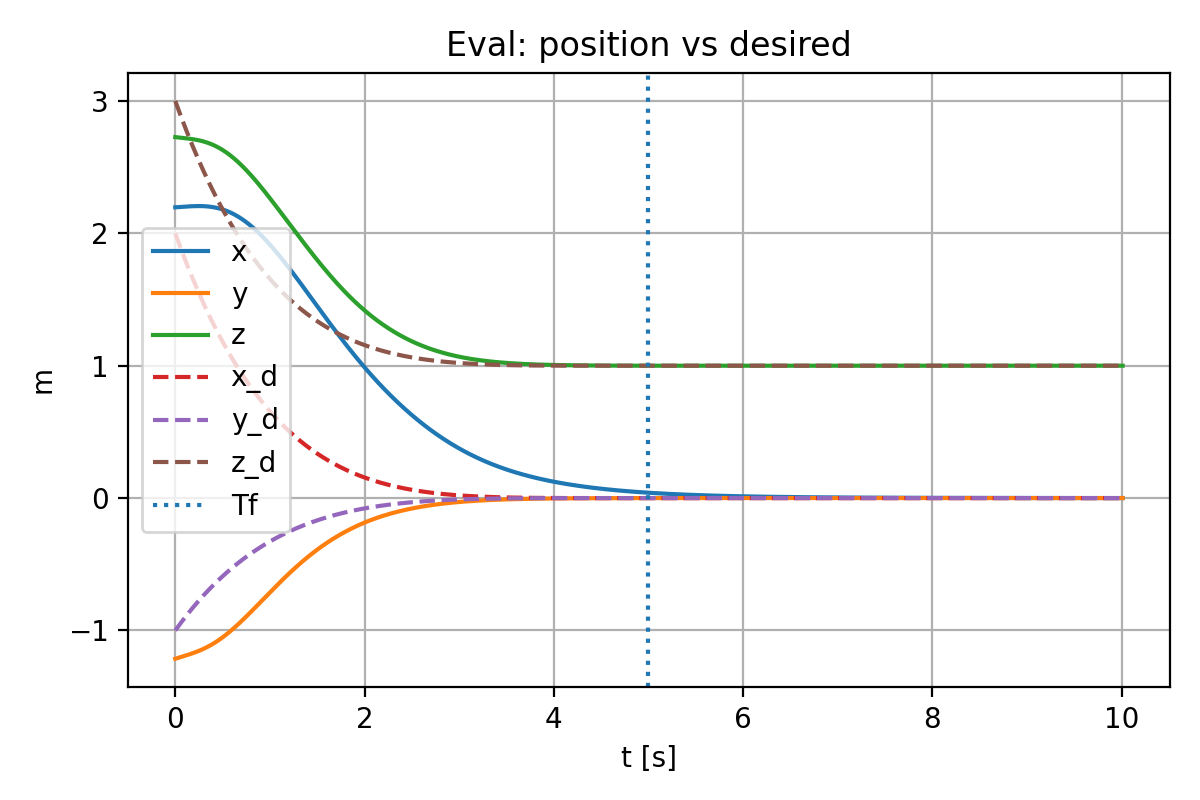}
\caption{Physical evaluation: inertial position versus desired position. The dotted line marks $T_f$; for $t>T_f$ the reference is held at $\mathbf{r}_d(T_f)$ and the quadcopter settles to hover.}
\label{fig:pos_vs_des}
\end{figure}

\begin{figure}[!h]
\centering
\includegraphics[width=0.99\linewidth]{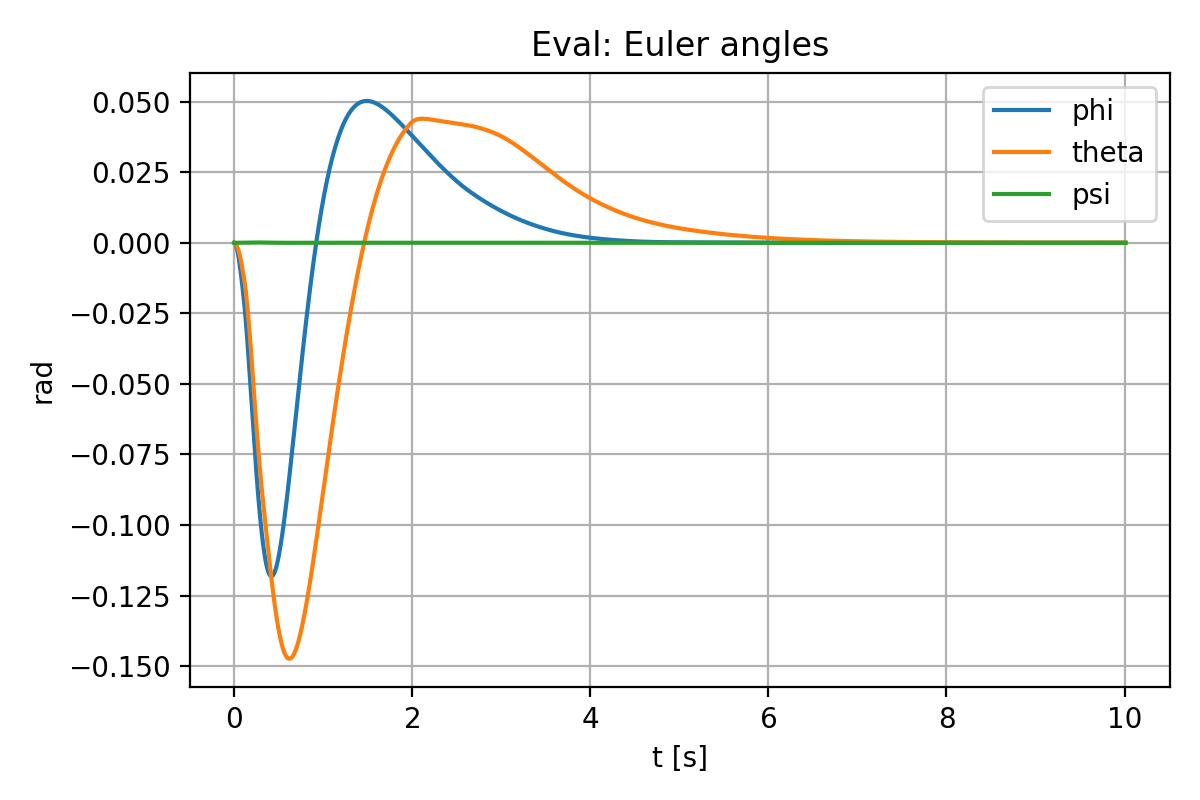}
\vspace{-1cm}
\caption{Physical evaluation: Euler angles $(\phi,\theta,\psi)$. Attitude excursions remain small and decay to near zero as tracking converges.}
\label{fig:euler}
\end{figure}

Figure~\ref{fig:dqn_gains} shows the translational feedback gains selected by the learned policy during a representative evaluation of shielded DQN rollout. Initially, the policy applies relatively aggressive gains to attenuate the large tracking errors. As the quadcopter approaches the desired trajectory and transitions into the post-$T_f$ hover regime, the controller progressively shifts to a milder gain configuration. This trend aligns with the purpose of online gain scheduling: strong corrective action during the transient phase and reduced aggressiveness near steady-state operation.

The corresponding external error-state trajectories are presented in Fig.~\ref{fig:dqn_states}. The position, velocity, acceleration, and jerk errors along all three translational axes decrease rapidly toward zero, while the yaw-related states remain bounded and well behaved. These results demonstrate stable closed-loop behavior under the learned scheduling policy and indicate that adaptive gain selection improves transient tracking performance without generating undesirable oscillatory responses.

\begin{figure}[h!]
\centering
\includegraphics[width=0.99\linewidth]{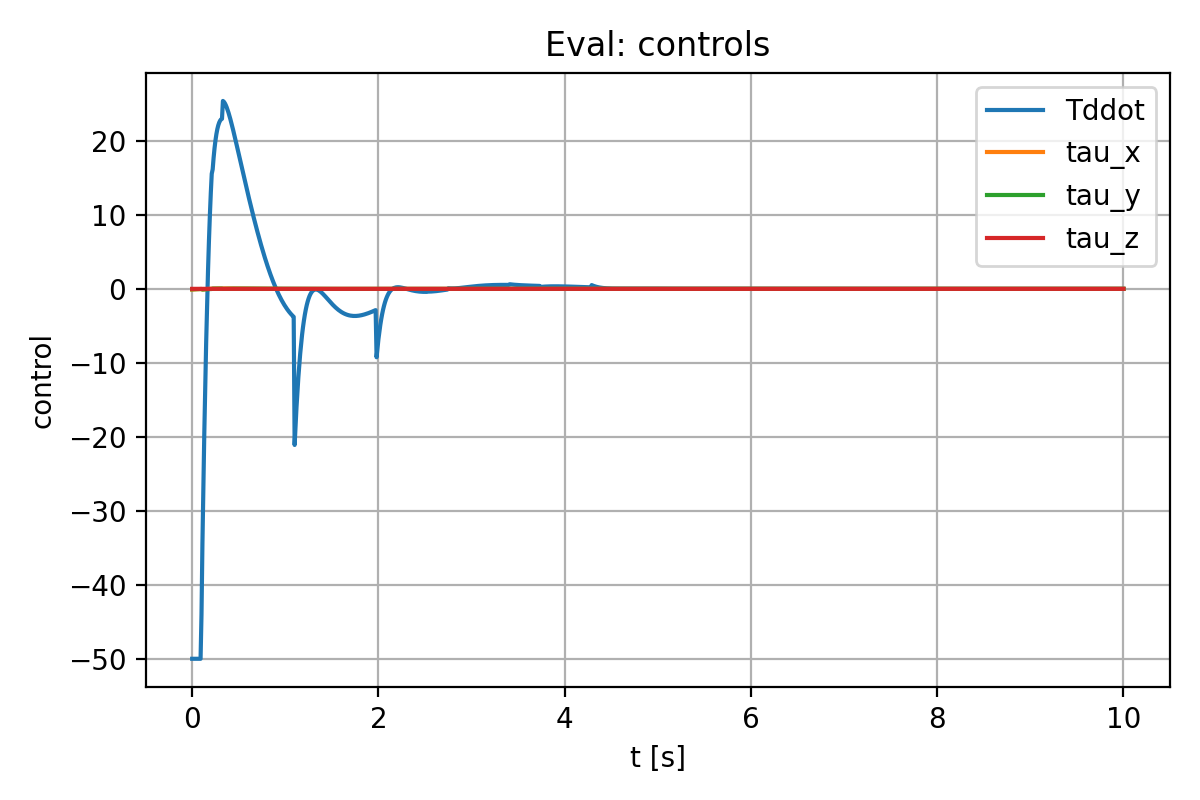}
\vspace{-1cm}
\caption{Physical evaluation: control inputs. The thrust second-derivative command $\ddot T$ and body torques $\boldsymbol{\tau}$ are largest during the initial transient and decrease as the state approaches the reference.}
\label{fig:controls}
\end{figure}

\begin{figure}[!h]
\centering
\includegraphics[width=0.99\linewidth]{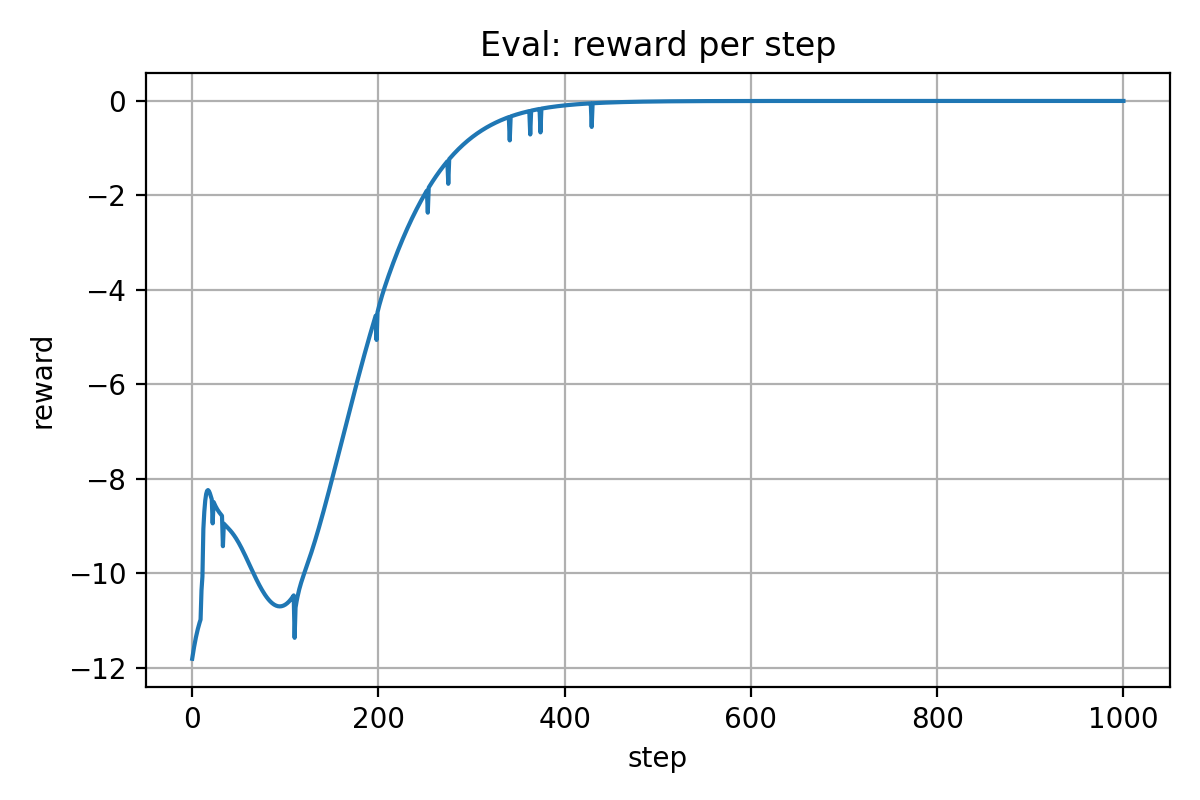}
\vspace{-1cm}
\caption{Physical evaluation: reward per step. The reward improves toward zero as tracking errors and control effort diminish over the episode.}
\label{fig:reward}
\end{figure}

To relate the gain evolution to the system motion, Fig.~\ref{fig:pos_vs_des} compares the quadcopter position components with the desired trajectory. The vehicle follows the quintic transition closely and settles smoothly to the final hover condition after $t=T_f$. Figure~\ref{fig:euler} shows the Euler-angle responses, indicating that the attitude deviations remain moderate throughout the maneuver and decay toward zero as the translational errors vanish. Hence, the improved trajectory-tracking performance is achieved without large rotational excursions.

Figure~\ref{fig:controls} depicts the control inputs, including the thrust second-derivative command $\ddot{T}$ and the body torques $\tau=[\tau_x,\tau_y,\tau_z]^\top$. As expected, the control effort is largest during the initial correction phase, when the tracking errors are greatest, and then diminishes as the vehicle converges to the desired hover equilibrium. Finally, Fig.~\ref{fig:reward} presents the per-step reward, which rises from strongly negative values during the early transient to values near zero as the state converges and the control effort decreases. Collectively, Figs.~1--6 show that the proposed gain-scheduled DQN achieves accurate trajectory tracking, limited attitude excursions, reduced control effort near convergence, and smooth transition to hover, while preserving the structured and stability-oriented nature of the underlying controller.

\section{Conclusion}\label{Conclusion}
The proposed framework yields a hybrid learning--control architecture in which reinforcement learning operates at a supervisory level over a structured family of stabilizing controllers. The control input is generated via model-based inversion, while learning selects gains from $\bar{\mathcal{K}}_{\mathrm{adm}}$. By restricting the action space to admissible gains, forward invariance of $\mathcal{X}_{\mathrm{adm}}$ is preserved throughout learning, ensuring safety by design. The admissible gain set $\mathcal{K}_{\mathrm{adm}}$ guarantees that all closed-loop trajectories remain within prescribed bounds, even during exploration, thereby decoupling safety from online optimization.

Within this safe envelope, the learning algorithm performs online gain scheduling to improve transient performance during the maneuver toward hover. This establishes a clear separation between safety and performance: safety is enforced by the controller structure, while learning optimizes tracking accuracy, control effort, and switching behavior over a finite set of physically meaningful controllers. The resulting formulation provides a control-oriented approach to reinforcement learning for nonlinear systems, improving interpretability and robustness. Future work will focus on extending the framework to uncertain and time-varying environments and scaling it to multi-agent systems under shared safety constraints.


\bibliography{ifacconf}

\end{document}